\documentclass[11pt,a4paper,UKenglish]{article}
\title{On the generic information capacity of relational schemas with a single binary relation} %TODO Please add

\author{
Beno\^{\i}t Groz\\
Universit\'{e} Paris-Saclay, CNRS, LISN, France\\
\texttt{groz@lisn.fr}\\[1em]
\begin{tabular}{cc}
Jan Hidders\thanks{Corresponding author.} & Nina Pardal \\
Birkbeck, University of London, UK & University of Edinburgh, UK \\
\texttt{j.hidders@bbk.ac.uk} & \texttt{npardal@ed.ac.uk}\\[1em]
Jan Van den Bussche & Piotr Wieczorek \\
Hasselt University, Belgium & University of Wrocław, Poland \\
\texttt{jan.vandenbussche@uhasselt.be} & \texttt{piotr.wieczorek@cs.uni.wroc.pl} \\
\end{tabular}
}
\usepackage{amsthm,amsmath,amssymb,hyperref,cleveref,thm-restate,enumerate}  % for theorem-like environments
\usepackage[skip=1em]{caption}
\usepackage{adjustbox}
\usepackage{lmodern}
\usepackage[T1]{fontenc}
\newtheoremstyle{remarkstyle}% <name>
  {3pt}% <Space above>
  {3pt}% <Space below>
  {}% <Body font>\usepackage{hyperref}
  {}% <Indent amount>
  {\bfseries}% <Theorem head font>
  {.}% <Punctuation after theorem head>
  { }% <Space after theorem head>
  {}% <Theorem head spec>

\theoremstyle{plain}
\newtheorem{theorem}{Theorem}[section]

\newtheorem{lemma}[theorem]{Lemma}
\newtheorem{proposition}[theorem]{Proposition}

\theoremstyle{remark}
\newtheorem{remark}[theorem]{Remark}
\newtheorem{example}[theorem]{Example}

\usepackage{todonotes}
\usepackage{arydshln} % for dashed lines in tables
\usepackage{stmaryrd}
\usepackage{hhline}
\usepackage{booktabs}

\newcommand{\key}[2]{{#1}^{\#}[{#2}]}

\newcommand{\ind}[4]{{#1}[{#2}] \subseteq {#3}[{#4}]}

\newcommand{\schema}{\mathcal{S}}
\newcommand{\sch}{\schema}
\newcommand{\V}{\mathcal{V}}
\newcommand{\cfun}{\mathcal{F}}
\newcommand{\sem}[1]{[\![ {#1} ]\!]}
\newcommand{\dom}{\mathbf{dom}}
\newcommand{\valDom}{\dom}
\newcommand{\idom}{\mathbf{idom}}

\newcommand{\Nat}{\mathbb{N}}

\newcommand{\adom}{\mathit{adom}}
\newcommand{\autgrp}[2]{\mathbf{Aut}_{#1}({#2})}
\newcommand{\id}[1]{\ensuremath{\mathsf{id}_{#1}}}
\newcommand{\idEq}{=_{\rm id}}
\newcommand{\ideq}{\idEq}
\newcommand\leqgen{\leq_{\rm gen}}
\newcommand\leqidgen{\leq^{\rm id}_{\rm gen}}
\newcommand\legen{<_{\rm gen}}
\newcommand\leabs{<_{\rm abs}}
\newcommand\leqabs{\leq_{\rm abs}}
\newcommand\substrong{\subset_{\rm strong}}
\newcommand\lestrong{<_{\rm strong}}
\newcommand\equivgen{\equiv_{\rm gen}}

\newcommand{\digraph}{\textsc{digraph}}
\newcommand{\symm}{\textsc{symm}}
\newcommand{\sourcefree}{\textsc{source~free}}
\newcommand{\sourcesinkfree}{\textsc{source\&sink~free}}
\newcommand{\sourcesinkfreebreak}{%
  \textsc{source}\allowbreak\&\textsc{sink~free}%
}
\newcommand{\cycles}{\textsc{cycles}}
\newcommand{\outdegleo}{\textsc{outdeg$\leq${\small 1}}}
\newcommand{\outdego}{\textsc{outdeg{\small 1}}}
\newcommand{\pathscycles}{\textsc{paths\&cycles}}
\newcommand{\symmdego}{\textsc{symm~deg{\small 1}}}

\begin{document}

\maketitle

\begin{abstract}

  We consider database schemas consisting of a single binary
  relation, with key constraints and inclusion dependencies.
  Over this space of 20 schemas, we completely characterize when
  one schema is generically dominated by another schema.  Generic
  dominance, a classical notion for measuring information
  capacity, expresses that every instance of a schema can be
  uniquely represented in the dominating schema, through
  application of a deterministic, generic data transformation.
  Our investigation is motivated both by current interest in
  schema design for graph databases, as well as by intrinsic
  scientific interest.  We also consider the ternary case, but
  without inclusion dependencies, and discuss how the notions
  change in the presence of object identifiers.

\end{abstract}

\section{Introduction}

The question of when two database schemas are equivalent, and
what that even means, is very natural.  Not surprisingly, this
question has been a classical research topic in database theory,
dating all the way back to Codd \cite{codd2}.  Rather than
approaching the notion of equivalence head on, it is more
informative to approach this question through a notion of
dominance,\footnote{Also called inclusion
\cite{atzeni-inclusion}.} in the sense of information capacity.
Here, one schema dominates another schema if there exists a
lossless transformation from instances of the second schema to
instances of the first.  Equivalence can then simply be defined
as dominance in both directions.

Of course, it still remains to formalize what we exactly mean by
a transformation, and what it means for a transformation to be
lossless.  In this paper, we focus on the notion of \emph{generic
dominance} \cite{hy_format}, where a transformation is defined as
a mapping from instances to instances that is \emph{generic}.
Lossless is taken to mean injective.

Genericity of a mapping from instances to instances, which means
that the mapping preserves isomorphisms, was identified early on
as a natural consistency criterion for database queries and
transformations, by Aho and Ullman \cite{au_universality}.  After
playing a fundamental role in the seminal work by Chandra and
Harel on computable queries, genericity became an established
notion in database theory \cite[Part~E]{ahv_book}.  Indeed, any
schema mapping defined in a logical language popular in database
theory (such as conjunctive queries, first-order logic, datalog,
source-to-target dependencies, fixpoint logic, second-order
logic) is generic, almost by definition of what it means to be
logical \cite{tarski_logical}.

\begin{figure}
  \centering
  \hspace*{1em}
  \begin{tabular}{@{}c@{}}
  \scalebox{0.6}{\includegraphics{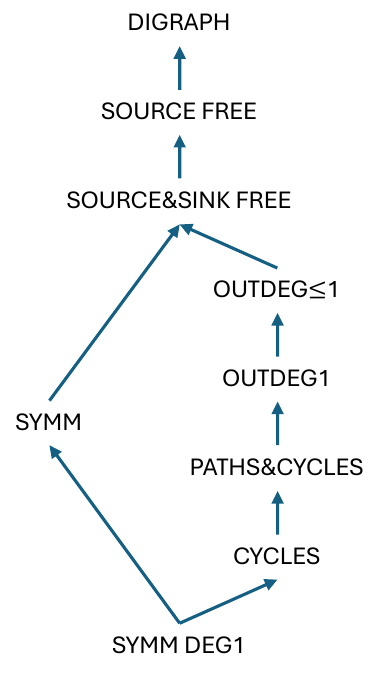}}
  \end{tabular}
  \hfill
  \begin{tabular}{@{}c@{}}
  \scalebox{0.80}{\includegraphics{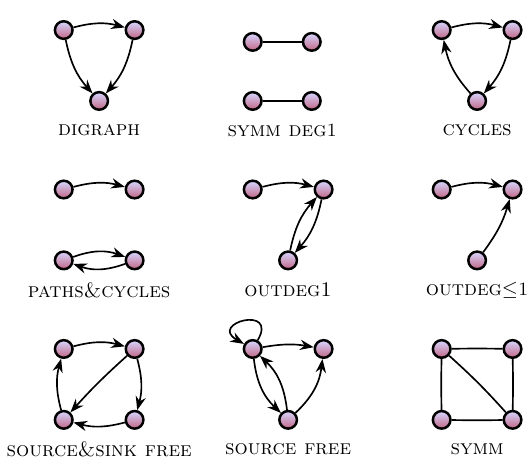}}
  \end{tabular}
  \hspace*{1em}
  \caption{Left: Generic dominance relationships over single binary
  relation schemes with keys and inclusion dependencies. Right:
  Conceptual illustration of the indicated graph classes. They are
  defined formally in Section~\ref{secbinary}.}
  \label{fighasse}
\end{figure}

\paragraph*{Results} In this paper our aim is to completely
characterize generic dominance of \textbf{single binary relation
schemes,} which can have integrity constraints taking the form of
\textbf{keys and inclusion dependencies.}  Syntactically, there
are 20 such schemas.  It turns out that they group into nine
distinct equivalence classes, each of which corresponds to a natural
class of directed graphs.  In Figure~\ref{fighasse} we map out
the \textbf{complete Hasse diagram} of strict generic dominance
relationships between these equivalence classes.  At the top is
\digraph, the class of all directed graphs (always assumed to be
without isolated nodes), corresponding to binary relations
without any integrity constraints.  At the bottom is the smallest
class $\symmdego$: symmetric (undirected) graphs where all nodes have degree
one.  In other words, these are just unions of isolated
undirected edges, and correspond to binary relations $R$ where
the first column is a key and the inclusion dependency
$R[1,2] \subseteq R[2,1]$ holds.

For example, from the diagram we see that the information
capacity of the class of source-free directed graphs
(\sourcefree) is not sufficient to represent all directed graphs.
Note that \sourcefree\ corresponds to binary relations $R$
satisfying $R[1] \subseteq R[2]$.
For another example, we see that the classes of undirected
graphs (\symm) and of disjoint unions of cycles (\cycles) have
incomparable information capacities: we cannot generically
transform disjoint unions of cycles into undirected graphs without loss of
information, or vice versa.  Note that \cycles\ corresponds to
binary relations $R$ where the first column is a key and where
$R[1] \subseteq R[2]$.  

These are just a few of the conclusions that can be drawn from
Figure~\ref{fighasse}; indeed, the figure packs quite a lot of results.
In a separate part of the paper, we also establish the
\textbf{Hasse diagram for ternary relations} with key
constraints, although we omit inclusion dependencies for that
case.

\paragraph*{Motivation}  We came to this research via the current
interest in schema design for graph databases (e.g.,
\cite{bonivoigt-graph-schema,pgschema,shacl-shex-pgschema,kger}).
In the context of \textbf{adoption} of graph database systems, an
important question arises: how can an existing operational
relational database schemas, with all its integrity constraints,
be viewed or represented as a graph database schema?  For such an
exercise to be successful, the graph schema should have at least
the information capacity of the relational schema.  Moreover,
since graph schema languages often have a close affinity to
conceptual models such as Entity-Relationship diagrams, the same
question arises, but in the other direction, in conceptual
modeling of databases \cite{enrico-conceptual}.

We are thus motivated to gain a comprehensive understanding of
the information capacities of graph schemas, and how they compare
to those of relational schemas.  Such a systematic investigation
is best initiated in the simplest nontrivial setting.  Now a
single binary relation scheme is simultaneously a very simple yet
nontrivial relational database schema, and the most stripped-down
graph schema possible, representing an unlabeled directed graph.
Moreover, keys and referential integrity (inclusion constraints)
are the two most basic kinds of integrity constraint for
relational database, and are also central in current proposals
(cited above) for graph schemas.  In this way we arrive exactly
at the setting of our work presented here.

\paragraph*{Organisation}
Section~\ref{secrelated} discusses
related work.  Section~\ref{secprelim} provides preliminaries.
In Section~\ref{secbinary} we introduce the nine directed graph
classes that result from generic equivalence classes of single
binary relation schemes with keys and inclusion dependencies.  We
then establish the Hasse diagram of generic dominance
relationships.
Section~\ref{secternary} does the same for the case of a single
ternary relation with key dependencies.

In Section~\ref{secid}, we present first insights into how the
notion of generic dominance should be extended to accommodate
transformations that \textbf{create new identifiers.}  In general, we may
need such identifiers when mapping arbitrary relational databases
to graph databases, or to restructure graph databases.  Finally,
in Section~\ref{seconc}, we provide a conclusion and an outlook
on the further research needed to extend our investigation to
general graph databases.

\section{Related work} \label{secrelated}

Initial concepts of dominance among database schemas were first
considered by Codd \cite{codd2} and Atzeni et
al.~\cite{atzeni-inclusion}, and consolidated by Hull and Yap in
their definitions of generic and absolute dominance
\cite{hy_format}.  Absolute dominance is based purely on
cardinality and is mostly useful for proving negative results,
since no absolute dominance implies no generic dominance. This
is a technique we will use here as well.  Hull and Yap were
motivated by the advent of semantic and complex-object data
models, involving data structures such as record formation, set
formation, and union types.  Generic dominance can be used to
gauge the information capacity of different such datatype
constructs \cite{abihull_restructtcs88}.  Another notable result
in this direction is that the class of directed graphs is not
generically dominated by the class of hypergraphs
\cite{odunlaing}.

The systematic investigation of generic dominance in the
relational model was initiated by Hull back in 1986
\cite{hull_siam}.  Hull showed how non-generic dominance can be
proven in specific cases where absolute dominance holds, by
careful analysis of automorphisms of instances.  This is also a
technique we will use here as well.  One of Hull's main technical
results is that every relational database schema is generically
dominated by a single relation scheme of a high enough arity (in
a sense that can be made precise).  From that, it is then shown
that relational database schemas without any constraints are
generically equivalent only if they are equal up to reordering.
Hull also defines \emph{calculous} dominance, which is generic
dominance through a data transformation expressible in
first-order logic (relational calculus).  Interestingly, in the
present paper, in each case where we have generic dominance, we
actually have calculous dominance.

Our present work is a clear continuation of Hull et al.'s systematic
line of research, which has lain dormant for 40 years, despite
its undeniable merit and relevance.  Note however that the
concept of information capacity in itself was certainly picked up
in a number of works and has been applied to validate schema
transformation rules (e.g., \cite{st1,st2,st3,st4,st5}).

We also note that in the early days of database theory,
equivalence of relational database schemes really was a notion of
equivalence of queries or views, namely, the equivalence of two
join expressions over the projections of a universal relation
\cite{bmsu_equivalence}.  Finally, we should mention the large
body of work on schema mappings in data exchange, which is
concerned with semantics, expressivity, and complexity of
logical mechanisms for specifying relationships between instances
of two different schemas \cite{foundations-data-exchange-book}.
In contrast, here we are concerned with the intrinsic question of
whether two given schemas are (generically) equivalent.

\section{Preliminaries} \label{secprelim}

We work in the relational data model under the unnamed
perspective \cite{ahv_book}.  A \emph{relational vocabulary} is a
finite set $\V$ of relation names, where each relation name has
an associate arity (a natural number).  We use the notation $R/k$
to indicate that $R$ has arity $k$.  We assume a fixed countably
infinite domain $\dom$ of atomic data values.  An \emph{instance}
$I$ of $\V$ assigns to each relation name $R/k$ of $\V$ a finite
$k$-ary relation $I(R)$ on $\dom$.
The set of elements from $\dom$ that appear in instance $I$ is
called the \emph{active domain} of $I$ and denoted by $\adom(I)$.

We will consider two kinds of integrity constraints on relational
vocabularies $\V$.  A \emph{key constraint} is an expression of
the form $\key RX$ with $R/k \in \V$ and $X \subseteq
\{1,\dots,k\}$.  An \emph{inclusion dependency (IND)} is an expression
of the form $\ind RXSY$, where $R/k$ and $S/m$ are relation names
in $\V$, and $X$ and $Y$ are lists of distinct elements of
$\{1,\dots,k\}$ and $\{1,\dots,m\}$ respectively, of the same
length.

An instance $I$ of $\V$ is said to satisfy the key
constraint $\key RX$ if $I(R)$ has no two different tuples that
agree on the components in $X$.  Also, $I$ is said to satisfy the
IND $\ind R{i_1,\dots,i_p}S{j_1,\dots,j_p}$ if
$\{(t_{i_1},\dots,t_{i_p}) \mid t \in I(R)\}$ is a subset of
$\{(t_{j_1},\dots,t_{j_p}) \mid t \in I(S)\}$.
Here, $t_i$ denotes the $i$th component of tuple $t$.

We refer to a pair $\schema=(\V,\Sigma)$, where $\V$ is a relational
vocabulary and $\Sigma$ is a set of keys and INDs
on $\V$, as a \emph{schema}.  An \emph{instance} of
$\schema$ is simply an instance of $\V$ satisfying all
constraints in $\Sigma$.  The set of all instances of $\schema$
is denoted by $\sem\schema$.  We can treat a relational
vocabulary as a schema where $\Sigma$ is empty.

\subparagraph{Schema mappings and genericity}
In this paper, a \emph{schema mapping} from
schema $\schema_1$ to schema $\schema_2$ is simply a total mapping
$f:\sem{\schema_1} \to \sem{\schema_2}$.  Such $f$ is called
\emph{generic} if there exists a finite set $C \subseteq \dom$
such that for every permutation $\pi$ of $\dom$ that is
the identity on $C$, and for every instance $I$ of $\schema_1$,
we have $f(\pi(I)) = \pi(f(I))$.  In this case we also say that
$f$ is $C$-generic.

\begin{remark}
  An equivalent definition of $C$-genericity is that it preserves
  $C$-isomorphism.  More precisely, if $\pi : I_1 \cong I_2$ is an
  isomorphism that is the identity on $C$, then also $\pi :
  f(I_1) \cong f(I_2)$. \qed
\end{remark}

The idea of $C$-genericity is that $f$ can treat the values in
$C$ as special constants, but otherwise treats all data values
symmetrically; only equalities matter.  For example, a
constant-equality selection $\sigma_{\$1=\textsf{`ICDT'}}(R)$ is
$\{\textsf{ICDT}\}$-generic.  The other operations of the
classical relational algebra (equality selection,
union, difference, equijoin, projection \cite{ullman})
are even $\emptyset$-generic.

\begin{example} \label{ex1}
  Consider the vocabulary $\V_1$ consisting of a number of unary relation
  names $S_1,\dots,S_k$, and the vocabulary $\V_2$ consisting of
  a single binary relation name $R$.  Pick $k$ distinct constants
  $c_1$, \dots, $c_k$ in $\dom$.  Consider the schema mapping $f$
  from $\V_1$ to $\V_2$ that maps an instance $I$ of $\V_1$ to
  the instance $J$ of $\V_2$ where $J(R)$ equals the result of
  evaluating the following relational algebra expression in $I$:
  \[ S_1 \times \{c_1\} \cup \cdots \cup S_k \times \{c_k\} \]
In words,
  every element is tagged with the identifiers of the unary relations
  it belongs to.  This mapping is $\{c_1,\dots,c_k\}$-generic.
  \qed
\end{example}

Generic mappings are ``data independent'': they make no
assumptions on the data values that can appear in the instances,
beyond the existence of a finite set of constants.  Computations
that do make such assumptions are typically not generic, as
illustrated next.

\begin{example}[\cite{hull_siam}]
  Consider the vocabulary $\V_2$ consisting of a single binary
  relation name $R$, and consider the
  vocabulary $\V_{11}$ consisting of a single unary relation $S$.
  Now assume $\dom$ equals the set of natural numbers, and
  consider the schema mapping $f$ from
  $\V_2$ to $\V_{11}$ that maps an instance $I$ of $\V_2$ 
  to the instance $J$ of $\V_{11}$ where $J(S) =
  \{2^x\cdot 3^y \mid (x,y) \in I(R)\}$.
  This mapping works only over the natural numbers and indeed
  is not generic.
\end{example}

\subparagraph{Notions of dominance}

We say that schema $\schema_1$ is \emph{generically dominated} by
schema $\schema_2$, denoted by $\schema_1 \leqgen \schema_2$, if
there exists an injective and generic schema mapping from
$\sch_1$ to $\sch_2$.  Injectivity ensures that we have a
lossless data transformation from instances of $\schema_1$ into
instances of $\schema_2$.  Genericity ensures that the
transformation is data independent.

\begin{example}
  In Example~\ref{ex1} we have $\V_1 \leqgen \V_2$ since the
  mapping exhibited there is injective. \qed
\end{example}

In order to define absolute dominance we first introduce the
\emph{cardinality function} $\cfun_\sch$ of a schema $\sch$.  For
any finite subset $D \subseteq \dom$, let $\sem{\sch}_D$ be the
set of instances $I$ of $\sch$ such that $\adom(I)\subseteq D$.
Then $\cfun_{\sch}: \Nat \to \Nat$ maps $n$ to the cardinality of
$\sem{\sch}_D$, for some $D$ of cardinality $n$.  This is
independent of the choice of $D$, since keys and INDs
are themselves $\emptyset$-generic.

We now say that $\schema_1$ is \emph{absolutely dominated} by
$\schema_2$, denoted by $\schema_1 \leqabs \schema_2$, if there
exists a number $c$ such that $\cfun_{\sch_1}(n) \leq
\cfun_{\sch_2}(n)$ for $n \geq c$.  Intuitively, for sufficiently
large active domains, there is enough ``room'' in $\sem{\schema_2}$ to
accommodate instances of $\schema_1$.

\begin{example}
  In Example~\ref{ex1} we do \emph{not} have $\V_2 \leqabs
  \V_{1}$, since $\cfun_{\V_2}(n)=2^{n^2}$ but
  $\cfun_{\V_{11}}(n)=2^{kn}$ and $k$ is fixed.
  We will see more interesting examples of no absolute dominance
  in Section~\ref{secbinary}.
\end{example}

If $f$ is $C$-generic, then $\adom(f(I)) \subseteq \adom(I)\cup
C$ \cite[Exercise~16.1]{ahv_book}.  Hence, generic dominance
implies absolute dominance.  Therefore, if we can prove that
there is no absolute dominance, we have also proven that there is
no generic dominance.  For examples of absolute dominance but no
generic dominance, see Hull \cite{hull_siam}.

We will often tacitly use the following property, which is
readily verified:

\begin{proposition}
  Generic dominance and absolute dominance are transitive.
\end{proposition}

\subparagraph{Notions of equivalence and strictness}

When $\sch_1 \leqgen \sch_2$ and vice versa, we naturally call
$\sch_1$ and $\sch_2$ \emph{generically equivalent} and denote
this by $\sch_1 \equivgen \sch_2$.   In contrast, when $\sch_1
\leqgen \sch_2$ but not vice versa, we say that $\sch_2$
\emph{strictly} generically dominates $\sch_1$ and denote this by
$\sch_1 \legen \sch_2$.
Strict absolute dominance is defined analogously and denoted by
$\leabs$.

  The well-known Schr\"oder-Bernstein theorem states that, if there exist
  injective mappings $f:A\to B$ and $g:B\to A$ for arbitrary
  (possibly infinite) sets $A$ and $B$, then there exists a
  bijection $h:A\to B$.  It turns out that a generic
  schema-mapping version of this theorem holds:

  \begin{theorem} \label{schroderbernstein}
If $\sch_1 \equivgen \sch_2$, then there exists a generic and
    bijective schema mapping from ${\sch_1}$ to ${\sch_2}$.
  \end{theorem}

    We prove this theorem in Appendix~\ref{appschroderbernstein}.
Note that the converse implication is immediate,
  since the inverse of a generic bijective schema
  mapping is generic as well.

\subparagraph{Logical implication and equivalence}

A trivial case of $\schema_1 \leqgen \schema_2$ happens when
$\schema_1$ logically implies $\schema_2$, i.e., when
$\sem{\schema_1} \subseteq \sem{\schema_2}$.
Likewise, $\schema_1 \equivgen \schema_2$ trivially holds when
$\schema_1$ and $\schema_2$ are logically equivalent, i.e., when
$\sem{\schema_1} = \sem{\schema_2}$.  We denote this by $\sch_1
\equiv \sch_2$.

\section{Schemas on a single binary relation name}
\label{secbinary}

Our goal in this section is to completely characterize generic
dominance between schemas on a single binary relation name $R$.
Syntactically there are 20 such schemas, as shown in
Table~\ref{tab19}.  The table already indicates the logical
equivalences that hold; these are readily verified.  The table
also indicates three more generic equivalences that hold; these
can be seen just by swapping the columns of $R$.

\begin{table} % JanVdB: changed to booktabs style
    \caption{Schemas on a single binary relation $R$.
    The descriptive labels indicating
    generic graph classes are explained in the text. These labels
    are only given for schemas that are not generically
    equivalent to a schema that comes earlier in the table.}
    \label{tab19}
    \centering
  \begin{adjustbox}{max width=\textwidth,center} % this shrinks the table to fit the page, but only if necessary
  \begin{tabular}{llll}
    \toprule
    schema & keys & INDs & generic graph class \\
    \midrule
 $\schema_0$ & none & none & \digraph \\
 $\schema_1$ & none & $\ind{R}{1}{R}{2}$ & \sourcefree \\
 $\schema_2 (\equivgen \schema_1)$ & none & $\ind{R}{2}{R}{1}$ & \\
 $\schema_3$ & none & $\ind{R}{1}{R}{2}$, $\ind{R}{2}{R}{1}$ &
    \sourcesinkfree \\
 $\schema_4$ & none & $\ind{R}{1, 2}{R}{2, 1}$ & \symm \\
 $\schema_5$ & $\key{R}{1}$ & none & \outdegleo \\
 $\schema_6$ & $\key{R}{1}$ & $\ind{R}{1}{R}{2}$ & \cycles \\
 $\schema_7$ & $\key{R}{1}$ & $\ind{R}{2}{R}{1}$ & \outdego \\
 $\schema_8 (\equiv \schema_{6})$ & $\key{R}{1}$ & $\ind{R}{1}{R}{2}$,
    $\ind{R}{2}{R}{1}$ & \\
 $\schema_9$ & $\key{R}{1}$ & $\ind{R}{1, 2}{R}{2, 1}$ &
    \symmdego \\
 %perfect matching (undirected)
 $\schema_{10} (\equivgen \schema_5)$ & $\key{R}{2}$ & none & \\
 %co-functional digraph
 $\schema_{11} (\equivgen \schema_7)$ & $\key{R}{2}$ &
    $\ind{R}{1}{R}{2}$ & \\
 $\schema_{12} (\equiv \schema_{6})$ & $\key{R}{2}$ &
    $\ind{R}{2}{R}{1}$ & \\
 $\schema_{13} (\equiv \schema_{6})$ & $\key{R}{2}$ &
    $\ind{R}{1}{R}{2}$, $\ind{R}{2}{R}{1}$ & \\
 $\schema_{14} (\equiv \schema_{9})$ & $\key{R}{2}$ & $\ind{R}{1, 2}{R}{2, 1}$
    & \\
 $\schema_{15}$ & $\key{R}{1}$, $\key{R}{2}$ & none &
    \pathscycles \\
 $\schema_{16} (\equiv \schema_{6})$ & $\key{R}{1}$, $\key{R}{2}$ &
    $\ind{R}{1}{R}{2}$ & \\
 $\schema_{17} (\equiv \schema_{6})$ & $\key{R}{1}$, $\key{R}{2}$ &
    $\ind{R}{2}{R}{1}$ & \\
 $\schema_{18} (\equiv \schema_{6})$ & $\key{R}{1}$, $\key{R}{2}$ &
    $\ind{R}{1}{R}{2}$, $\ind{R}{2}{R}{1}$ & \\
 $\schema_{19} (\equiv \schema_{9})$ & $\key{R}{1}$, $\key{R}{2}$ &
    $\ind{R}{1, 2}{R}{2, 1}$ & \\
    \bottomrule
\end{tabular}
\end{adjustbox}
\end{table}

Since a binary relation is a directed graph (without isolated
nodes), the instances of a schema form a class of directed
graphs.  The nine classes that result from the table
can be described using standard graph terminology:
\begin{itemize}
  \item \digraph, all directed graphs (no constraints).
  \item \sourcefree, no source nodes.
  \item \sourcesinkfree, no source or sink nodes.
  \item \symm, edges are symmetric (undirected graphs).
  \item \outdegleo, no node can have more than one outgoing edge.
  \item \cycles, disjoint union of cycles.
  \item \outdego, each node has exactly one outgoing edge.
  \item \symmdego, union of isolated undirected edges (a perfect
    matching).
  \item \pathscycles, disjoint union of directed paths and
    cycles.
\end{itemize}

We will establish:

\begin{theorem} \label{theorbinary}
  There are no other generic equivalences on schemas
  $\sch_1$ to $\sch_{19}$ than those implied in Table~\ref{tab19}.
  Moreover, the partial order resulting from generic dominance 
  over the nine equivalence classes, is exactly as shown in the
  Hasse diagram of Figure~\ref{fighasse}.
\end{theorem}

The proof of this theorem is threefold.  First, we must prove all generic
dominances shown as edges in the Hasse diagram.  Second, we must
prove that these dominances are strict, i.e., that there are no
generic dominances on the inverse edges.  Third, for any two
classes shown as incomparable in the diagram, we must show that
in both directions there is no generic dominance.

\paragraph*{Dominance}

To show generic dominance along the edges of the diagram, 
we only have to deal with $\outdegleo \to \sourcesinkfree$.
Indeed, for all other edges, logical implication holds, so
generic dominance is immediate.

So we prove $\outdegleo \leqgen \sourcesinkfree$.  Consider the
mapping $m$ that maps an instance $I$ of $\outdegleo$ to the instance $J$ of
$\sourcesinkfreebreak$ where $J(R)$ equals the result of the following
relational algebra expression evaluated in $I$:
\[ R \cup \pi_{1,1}(\pi_1(R)-\pi_2(R)) \cup
(\pi_2(R)-\pi_1(R)) \times (\pi_1(R)-\pi_2(R)) \]
In words, $m$ adds a self-loop to all sources of $I$, and adds
edges from all sinks of $I$ to all sources of $I$.  Denote by
$e_1$ the expression that adds the self-loops, and by $e_2$ the
expression that adds the edges from sinks to sources.

Since there is only one relation name $R$, we can
abbreviate $I(R)$ by $I$ and $m(I)(R)$ by $m(I)$.
Since $I$, $e_1(I)$, and $e_2(I)$ are mutually disjoint
and $m(I)$ equals their union, to show injectivity of $m$ it
suffices to show that $e_1(I)$ and $e_2(I)$ can be uniquely
determined from $m(I)$.

For $e_1(I)$ this is clear, because the nodes on self-loops that
were already present in $I$ are neither sources nor sinks. Hence,
by outdegree${} \leq 1$, these nodes do not have any other
outgoing edges in $I$, and neither in $m(I)$.  In contrast, the
nodes on self-loops added by $e_1$ do have another outgoing edge
in $m(I)$.

Having identified the original sources from $I$ in $m(I)$,
in order to determine $e_2(I)$, it
remains to identify the original sinks.  Let $(y,x) \in m(I)$
with $x$ a source in $I$.  Then $(y,x) \notin I$, so $(y,x)$
was added by $e_2$. Hence, the sinks of $I$ are precisely those 
nodes that have an edge in $m(I)$ to a source in $I$.

\paragraph*{Strong inclusion, dominance, and strictness}

To show that the generic dominances along the edges of the
diagram are strict, we first consider all edges except
$\outdegleo \to \sourcesinkfree$.  For each such edge $\sch \to
\sch'$ we already noted that $\sem{\sch}$ is included in
$\sem{\sch'}$, but there is more: for any sufficiently large
finite domain $D \subseteq \dom$, we have the strict inclusion
$\sem{\sch}_D \subset \sem{\sch'}_D$.  We call this 
dominance relationship \emph{strong inclusion} and denote it by
$\sch \substrong \sch'$.

For our proofs we also need a strict version of absolute
dominance.  We say that $\sch$ is \emph{strongly dominated} by
$\sch'$, denoted by $\sch \lestrong \sch'$, if there exists a
number $c$ such that $\cfun_{\sch}(n) < \cfun_{\sch'}(n)$ for $n
\geq c$.  We note three immediate observations:

\begin{lemma} \label{lemstrong}
  \begin{enumerate}[(i)]
    \item
  Strong inclusion implies strong dominance.
\item
  Strong dominance implies absolute dominance.
\item
  Strong dominance is transitive.
\item
  If $\sch \lestrong \sch'$ then $\sch' \leqabs \sch$ does not hold.
  \end{enumerate}
\end{lemma}

From (i) and (iv) above we can already conclude that all edges
in the diagram, except for $\outdegleo \to \sourcesinkfree$,
represent a strict dominance.

It thus remains to prove that $\sourcesinkfree \leqgen
\outdegleo$ does not hold.  Thereto we will prove the following about
the class \symm:

\begin{lemma} \label{lemf4f5}
\symm\ strongly dominates \outdegleo.
\end{lemma}

This lemma gives what we want.
Indeed, we also have $\symm \substrong \sourcesinkfree$.  Then by
transitivity, \[\outdegleo \lestrong \sourcesinkfree,\] so that
\outdegleo\ does not absolutely, whence not generically, dominate
$\sourcesinkfree$, as desired.

It now remains to prove Lemma~\ref{lemf4f5}.  In order to prove
this and other strong domination results, we made the following
effort:

\begin{proposition} \label{propcard}
  Table~\ref{tabcard} contains the cardinality functions for the
  nine generic graph classes of Table~\ref{tab19}.
\end{proposition}

\begin{table}
\caption{Cardinality functions for the schemas from
  Table~\ref{tab19}.  OEIS refers to the corresponding entry in
  the on-line encyclopedia of integer sequences \cite{oeis}.  To
  date, it appears there is no entry in OEIS corresponding directly to
  $\cfun_3$.
In $\cfun_3$, $\mathbf{a}_k$ equals the number of $k \times
    k$~boolean matrices with no zero-row nor zero-column or,
    equivalently, the number of edge covers in the complete
    bipartite graph $\mathbf{K}_{k, k}$.
    In $\cfun_9$, 
$\mathbf{I}_k$ equals the number of involutions
    on a set with $k$ elements.}
  \label{tabcard}
\centering
\begin{adjustbox}{max width=\textwidth,center} % this shrinks the table to fit the page, but only if necessary
\renewcommand{\arraystretch}{1.45}
\begin{tabular}{llll} % JanVdB: converted to booktabs style
\toprule
schema & generic graph class & cardinality function & OEIS \\
\midrule
$\schema_0$ & \digraph & $\cfun_0 : n \mapsto 2^{n^2}$ & A002416\\
$\schema_1$ & \sourcefree &
    $\cfun_1 : n \mapsto \sum_{k=0}^n \binom{n}{k} (2^{k}-1)^k$ & A251182 \\
$\schema_3$ & \sourcesinkfree &
    $\cfun_3 : n \mapsto \sum_{k=0}^n \binom{n}{k} \mathbf{a}_k$ & \\

$\schema_4$ & \symm & $\cfun_4 : n \mapsto 2^{\binom{n}{2} + n}$ & A006125 \\
$\schema_5$ & \outdegleo &
$\cfun_5 : n \mapsto \sum_{k=0}^n \binom{n}{k} n^k = (n + 1)^n$ & A000169\\
  $\schema_6$ & \cycles &
  $\cfun_{6} : n \mapsto \sum_{k=0}^n \binom{n}{k} k!$ & A000522
  \\
$\schema_7$ & \outdego &
    $\cfun_7 : n \mapsto \sum_{k=0}^n \binom{n}{k} k^k$ & A086331 \\
$\schema_{9}$ & \symmdego &
    $\cfun_{9} : n \mapsto \sum_{k=0}^n \binom{n}{k} \mathbf{I}_k$ & A000085\\
$\schema_{15}$ & \pathscycles &
    $\cfun_{15} : n \mapsto \sum_{k=0}^n \binom{n}{k}^2 k!$ & A002720 \\
    \bottomrule
\end{tabular}
\end{adjustbox}
\end{table}

The proof of the above Proposition is in Appendix~\ref{appcard}.
Recall that all schema equivalences in Table~\ref{tab19} are
either logical equivalence or symmetry by swapping the two columns of
$R$.  Since these do not change the cardinality function,
Table~\ref{tabcard} effectively lists the cardinality functions 
of all schemas from Table~\ref{tab19}.

With the concrete cardinality functions in hand, we can now
verify Lemma~\ref{lemf4f5}.  Indeed, from the expressions we see
that $\cfun_4(n)/\cfun_5(n) \geq 2^{n(n+1)/2 - n\log_2 (n+1)}$
from which the claim follows.

\paragraph*{Incomparable pairs}

The pairs of schemas indicated as incomparable in the Hasse
diagram are \symm\ with any of the four classes on the chain
$\cycles \to \pathscycles \to \outdego \to \outdegleo$, and vice
versa.  To show that neither of the four classes on this chain
generically dominates \symm, it suffices to do this for the
largest one, \outdegleo.  However, this already follows from
Lemma~\ref{lemf4f5}.

To show, conversely, that \symm\ does not generically dominate
any of the four classes on the chain, it suffices to do this for
the smallest class, \cycles.  Note that Lemma~\ref{lemf4f5}
implies that \symm\ strongly dominates \cycles.  Yet, we can show:

\begin{lemma}
  \symm\ does not generically dominate \cycles.
\end{lemma}
\begin{proof}
  As we did earlier, for instances $I$, we abbreviate $I(R)$ by $I$.  
  Assume, for the sake of contradiction, that there exists a
  $C$-generic injective schema mapping $m$ from \cycles\ to
  \symm.  Let $I$ be the instance $\{(a,b),(b,c),(c,a)\}$ where
  $a,b,c \notin C$.  Note that the cyclic permutation $\sigma=(a\ b\ c)$
  is a $C$-automorphism of $I$ (meaning that $\sigma$ is the
  identity on $C$ and that $\sigma(I)=I$).  However, the transposition
  $\tau=(b \ c)$ is not an automorphism of $I$.

  Let $J=m(I)$.  Since $m$ is $C$-generic and injective, $I$ and
  $J$ have exactly the same $C$-automorphisms.  So, if we can
  show that $\tau$ is an automorphism of $J$, we have
  arrived at the desired contradiction.  Note that $\sigma$ 
  is an automorphism of $J$ because it is an automorphism of $I$.

  We begin by showing $\tau(J) \subseteq J$.  Let $(x,y) \in J$.
  Since $m$ is $C$-generic, $x$ and $y$ must belong to $\{a,b,c\}
  \cup C$.  We consider the following possibilities for the pair
  $(x,y)$.
  \begin{enumerate}
    \item
      $x$ and $y$ are distinct and both belong to $\{a,b,c\}$.
      Then $\sigma(J)\subseteq J$ and the symmetry of $J$ imply that $J$
      contains \emph{all} such pairs.  Hence, $\tau(x,y) \in J$.
    \item
      $x=y$ and $x \in \{a,b,c\}$.  Then
      $\sigma(J)\subseteq J$ implies that again $J$ contains all
      such pairs.  Hence, $\tau(x,y) \in J$.
    \item
      One of $x$ and $y$ is in $C$, and the other is in
      $\{a,b,c\}$.  Again, 
      $\sigma(J)\subseteq J$ and the symmetry of $J$ imply that $J$
      contains all such pairs, so $\tau(x,y) \in J$.
    \item
      $x$ and $y$ are both in $C$.  Then $\tau(x,y)=(x,y)$.
  \end{enumerate}

  We can follow the exact same analysis to show that also
  conversely, $J \subseteq \tau(J)$.
  We conclude that $\tau$ is an automorphism of $J$, as desired.
\end{proof}

\section{The ternary case} \label{secternary}

\newcommand\T{\mathcal T}

In this section, we push things a little further by moving
to schemas on a single \emph{ternary} relation
name $R$.  However, to keep things tractable, we only consider
key constraints, as adding INDs would lead to hundreds of schemas
to consider.  Since we only have keys, and a key is a subset of
the columns $\{1,2,3\}$, a schema $\T$ in this setting
is essentially a subset of the powerset of $\{1,2,3\}$.
To avoid logical equivalences, we can restrict attention to
Sperner families (i.e., no set in $\T$ contains another set in $\T$).
Moreover, we do not list a schema if it is symmetric (through a
permutation of the columns) to a schema already listed earlier.
This leads us to eight schemas, shown in Table~\ref{tab3}.

\begin{table}
  \caption{Schemas with only key constraints over a single ternary
  relation name $R$.}
  \label{tab3}
  \centering
  \begin{tabular}{ll}
    \toprule
    schema & keys \\
    \midrule
    $\T_0$ & none \\
    $\T_1$ & $\{1\}$ \\
    $\T_2$ & $\{1\}$, $\{2\}$ \\
    $\T_3$ & $\{1\}$, $\{2\}$, $\{3\}$ \\
    $\T_4$ & $\{1,2\}$ \\
    $\T_5$ & $\{1,2\}$, $\{1,3\}$ \\
    $\T_6$ & $\{1,2\}$, $\{1,3\}$, $\{2,3\}$ \\
    $\T_7$ & $\{1\}$, $\{2,3\}$ \\
    \bottomrule
  \end{tabular}
\end{table}

We note the following two chains of strong inclusions:
\[ \T_3 \substrong \T_2 \substrong \T_7 \substrong \T_1 \quad
\text{and} \quad
\T_6 \substrong \T_5 \substrong \T_4 \substrong \T_0 \]
Hence, if we can prove that $\T_1 \leqgen \T_6$ but not vice versa,
we have established that the Hasse diagram is just one long
chain:
\begin{theorem}
  The Hasse diagram for generic dominance over the schemas in
  Table~\ref{tab3} is as follows:
  \[ \T_3 \to \T_2 \to \T_7 \to \T_1 \to \T_6 \to \T_5 \to \T_4
  \to \T_0
  \]
\end{theorem}

We prove the theorem by the following two lemmas.

\begin{lemma}
  $\T_1 \lestrong \T_6$.
\end{lemma}
\begin{proof}
  We begin by observing that \[ \cfun_{\T_1}(n) =
  \sum_{k = 0}^{n} \binom{n}{k} {n^{2k}} = (1 + n^2)^n. \]
  Indeed, we can select $k$ key values for the first column, and
  then for each we select a pair of values from the domain, for
  which there are $n^2$ distinct options. 

  We further \textbf{claim}
  that \[ \cfun_{\T_6}(n) \geq \frac{(n!)^{2n}}{n^{n^2}}
  \qquad \text{for $n>0$.} \]
  Assuming this claim, the lemma follows as we can use Stirling's
  formula to obtain $\cfun^3_{\T_6}(n)\in \Omega((n/e^2)^{n^2})$.
  The latter indeed strongly dominates $\cfun_{\T_1}(n)$.

  The claim above is based on the observation that the number of
  instances of $\T_6$ on a domain of at most $n$ elements is
  equal to the number of \emph{partial Latin squares} of order
  $n$. This can be seen as follows. Every filled position
  corresponds to a tuple in the instance where the row and the
  column identify the first two values in the tuple and the
  number in the entry is the third value in the tuple. Clearly the
  row and column identify the entry's value, and since in a Latin
  square a number can appear at most once in a certain row or
  column, it holds that given an entry value, the row identifies
  the column of the entry and vice versa.

  Now the formula in the
  claim is indeed a lower bound for the number of (full) Latin
  squares \cite[page 157]{vanlint-wilson},
  and thus also a lower bound of the number of
  partial Latin squares.
\end{proof}

\begin{lemma} \label{lemtriples}
$\T_1 \leqgen \T_6$.
\end{lemma}
\begin{proof}
  As done earlier, since there is only one relation name $R$, for
  instances $I$ we abbreviate $I(R)$ by $I$.  Let $I$ be any
  instance of $\T_1$.  Since the first column is a key in $I$, we
  can view $I$ as two mappings $I_2$ and $I_3$ from the first
  column of $I$ to $\adom(I)$, so that \[ I = \{(x,I_2(x),I_3(x))
  \mid x \in \pi_1(R)(I)\}. \]

  The idea of a mapping from $\T_1$ to $\T_6$ is to try to
  represent, for each $x$, the
  triple $(x,y,z)$ in $I$ by the two triples $(x,y,x)$ and
  $(x,x,z)$. For example:
  \[ \begin{array}{|ccc|} \hline a_1 & b & c \\ a_2 & b & c \\
  \hline \end{array} \quad \mapsto \quad
  \begin{array}{|ccc|} \hline a_1 & b & a_1 \\
    a_1 & a_1 & c \\
  a_2 & b & a_2 \\
    a_2 & a_2 & c \\
  \hline \end{array} \]
  Note that the instance on the left is an instance of $\T_1$ but
  not of $\T_6$ since $\{2,3\}$ is not a key.  The instance on
  the right satisfies all constraints of $\T_6$, and the instance
  on the left can be recovered from it, by looking at equality
  patterns.

  This mapping needs refinement, however, when there are
  equalities already present in the triple $(x,y,z)$ in $I$ to be
  represented.  For example, if $x=y$, we would map to
  $\{(x,x,x),(x,x,z)\}$ thus violating $\T_6$ since $\{1,2\}$ is
  not a key.  We will avoid this violation by not producing
  $(x,x,x)$ in such cases.

  Another need for refinement occurs when, for a triple $(x,y,z)$
  in $I$, we have $x = I_3(y)$, for example:
  \[ \begin{array}{|ccc|} \hline a & b & c \\ b & d & a \\
  \hline \end{array} \quad \mapsto \quad
  \begin{array}{|ccc|} \hline a & b & a \\ a & a & c \\
  b & d & b \\ b & b & a \\ \hline \end{array} \]
  The instance on the right now violates the key $\{2,3\}$.  We will
  repair this by using an alternative representation of $(x,y,z)$
  in this case by the two triples $(x,y,y)$ and $(x,x,z)$.

  The formal definition of the mapping and the verifications that
  it correctly maps $\T_1$ to $\T_6$ and is injective are
  given in Appendix~\ref{apptriples}.  The mapping is also
  $\emptyset$-generic and is actually expressible in relational
  algebra.
\end{proof}

\section{Creation of identifiers} \label{secid}

In practical applications of schema mappings, we may want the schema
mapping to be able to introduce new identifiers.  

\begin{example} \label{exid}
  It is well known
  how a $k$-ary relation $R$ can be represented by
  $k$ binary relations $A_1$, \dots, $A_k$, by introducing tuple
  identifiers. For example:
  \[ \begin{array}{|ccc|}
    \multicolumn{3}{c}{R} \\
    \hline a & b & c \\ d & e & f \\
  \hline \end{array} \quad \mapsto \quad
  \begin{array}{|cc|} 
    \multicolumn{2}{c}{A_1} \\
    \hline 
    \id 1 & a \\ \id 2 & d \\
  \hline \end{array} \quad
  \begin{array}{|cc|} 
    \multicolumn{2}{c}{A_2} \\
    \hline 
    \id 1 & b \\ \id 2 & e \\
  \hline \end{array} \quad
  \begin{array}{|cc|} 
    \multicolumn{2}{c}{A_2} \\
    \hline 
    \id 1 & c \\ \id 2 & f \\
  \hline \end{array} \]
  \qed
\end{example}

  Indeed, creation of identifiers is a standard feature in schema
  mapping languages and is known by many names (existential
  rules, Skolem functions, blank node generation, etc.)
  It is also easily accomplished in SQL by auto-incrementing.
  Actually, value invention in query languages was already
  investigated intensively in the 1980s and 1990s
  \cite[Part~E]{ahv_book}.

  The notion of generic dominance, however, is
  language-independent, and based instead on the general notion of
  generic schema mapping.  To discuss information capacity in the
  setting of identifier creation, we thus need to extend the
  notion of generic mapping.  Indeed, the mapping illustrated in the
  example above is not generic in the sense defined in
  Section~\ref{secprelim}, yet we may want to allow it.

  Interestingly, such extensions of genericity have already been
  proposed \cite{ak_iql,vvag_compl}.  In this section, we will
  adopt the proposal by Abiteboul and Kanellakis 
  \cite{ak_iql}. Before we can do this, however, we must first
  extend our relational data model to accommodate identifiers.

\subsection{Schemas and instances with identifiers}

In relational vocabularies,
instead of a single arity, we now assume that each relation name
$R$ has a double arity $(k,\ell)$, with $k$ and $\ell$ natural
numbers.  The idea is that relation instances of $R$ will have $k$
identifier columns and $\ell$ value columns.  

Next to the domain $\dom$ of data values, we now assume a
separate and disjoint domain $\idom$ of identifiers.  An
instance $I$ of a relational vocabulary $\V$ now assigns to each
relation name $R/(k,\ell)$ in $\V$ a finite subset of $\idom^k
\times \dom^\ell$.

\begin{example} 
  In Example~\ref{exid} we showed an instance of the vocabulary
  consisting of three relation names $A_1$, $A_2$ and $A_3$, each
  of arity $(1,1)$. \qed
\end{example}

When two instances differ only in the choice of identifiers, we
do not want to distinguish them.  Formally, $I_1$ and $I_2$ are
called \emph{equal up to ids}, denoted by $I_1 \idEq I_2$, if
there is a permutation $\pi$ of $\idom$ such that $\pi(I_1)=I_2$.
Here, $\pi$ does not touch the data values from $\dom$.

\subsection{Generic dominance in the presence of identifiers}
\label{secidgen}

A \emph{schema}, as before, consists of a relational vocabulary
(now allowing identifier columns) accompanied with a set of
integrity constraints.  For the following definitions, it does
not matter what form the integrity constraints take.  As long as
we have a notion of when an instance satisfies an
integrity constraint, we can define the set $\sem{\sch}$ of
instances of schema $\sch$ as before as the set of instances of
the vocabulary that satisfy all constraints.

In order to allow identifier creation, we can no longer define a
schema mapping as a function, since the choice of identifiers is
arbitrary.  Still, we want the mapping to be functional up to
the choice of newly created identifiers.\footnote{Functionality
up to the choice of ids was
called \emph{determinacy} by Abiteboul and Kanellakis
\cite{ak_iql}.}

Formally, a schema mapping from $\sch$ to $\sch'$ is now a
relationship $M \subseteq \sem{\sch} \times \sem{\sch'}$.  We
consider the following properties for such $M$:

\begin{description}
    \item[Total:] for each  $I \in \sem{\schema}$ there
      exists $J \in \sem{\schema'}$ such that $M(I, J)$
    \item[Functional up to ids:] if $M(I,J)$ and $M(I,J')$, then
      $J \idEq J'$ (equal up to ids).
    \item[Injective up to ids:] if $M(I,J)$ and $M(I',J)$, then
      $I \idEq I'$.
\item[Generic for values:] There exists is a finite set $C \subseteq
  \valDom$ such that for any permutation $\pi$ of $\dom$ that is
    the identify on $C$, and for any $(I,J) \in M$, also
    $(\pi(I),\pi(J)) \in M$. (Here, $\pi$ does not touch the
    identifiers.)
\item[Generic for identifiers:] For
any permutation $\pi$ of $\idom$,
    and for any $(I,J) \in M$, also $(\pi(I),\pi(J)) \in M$.
    (Here, $\pi$ does not touch the data values.)
\end{description}

We can finally define: \emph{$\sch$ is generically dominated by
$\sch'$ with identifiers}, denoted by $\sch \leqidgen \sch'$, if
there exists a schema mapping from $\sch$ to $\sch'$ that has all
the above five properties.

\begin{remark} \label{remeqid} A mapping as above that is
  functional up to ids, and generic for identifiers,
  can also be viewed as a function that
  maps equality classes up to ids to equality classes up to ids.
  Genericity then needs only to be specified for permutations of
  $\dom$.  Also, injectivity up
  to ids then reverts to standard injectivity.
\qed
\end{remark}

Generic equivalence with identifiers can again be defined as
$\leqidgen$ in both directions.  Like
Theorem~\ref{schroderbernstein}, also in this extended context,
we can show the following Schr\"oder-Bernstein-like theorem:
\begin{theorem}
  $\sch$ is generically equivalent to $\sch'$ with identifiers if
  and only if $\sch \leqidgen \sch'$ by a mapping that is also
  surjective.
\end{theorem}
We omit the proof.  Surjectivity means here that for each $J \in
\sem{\sch'}$ there exists $I \in \sem{\sch}$ such that $M(I,J)$.

Also the notions of \textbf{cardinality functions} and
\textbf{absolute} and \textbf{strong dominance} can be
extended with identifiers.  Here we must take care not to count
individual instances, but to count modulo equality up to ids.
We omit the formal details.  The notion of \textbf{strong
inclusion} remains largely the same.

Finally, the notion of \textbf{isomorphism} can be adapted: an 
\emph{isomorphism with identifiers} from instance $I$ to instance
$J$ is a pair $(\iota,\pi)$, where $\iota$ ($\pi$) is a
permutation of $\idom$ ($\dom$) such that $J = \iota(\pi(I))$.
However, if we take the perspective of equality classes up to ids
(Remark~\ref{remeqid}) we only need to consider permutations
of $\dom$.

\subsection{Schemas with a binary relation with one identifier column}

As an exercise in working with the identifier notions just introduced,
let us explore again the simplest nontrivial setting:
that of a single binary relation name $R$ with one identifier
column and one value column.  So formally, $R$ has arity $(1,1)$.
As integrity constraints on $R$ we only consider
keys.  Thus, there are just four schemas to consider:
\newcommand\K{\mathcal K}
\newcommand\Knone{\K_{\rm none}}
\newcommand\Kone{\K_{\rm id}}
\newcommand\Ktwo{\K_{\rm val}}
\newcommand\Kboth{\K_{\rm both}}
\begin{description}
  \item[$\Knone$:] no keys
  \item[$\Kone$:] the first (identifier) column is a key
  \item[$\Ktwo$:] the second (value) column is a key
  \item[$\Kboth$:] both columns are keys
\end{description}
We have two strong inclusion chains:
\[ \Kboth \substrong \Kone \substrong \Knone \quad \text{and}
\quad
\Kboth \substrong \Ktwo \substrong \Knone \]
So, we only need to understand how $\Kone$ and $\Ktwo$ compare
with respect to generic dominance with identifiers.  We can show
the following.  The proof shows that the presence of identifiers 
has an intricate impact on the information capacity of schemas.

\begin{theorem}
  $\Kone$ and $\Ktwo$ are incomparable with respect to generic
  dominance with identifiers.
\end{theorem}
\begin{proof}
  First, assume, for the sake of contradiction, that
  $\Kone \leqidgen \Ktwo$ by schema mapping $M$.
  Then $M$ is $C$-generic for values for some finite $C \subseteq
  \dom$.  If $M(I,J)$, then $\adom(J) \subseteq
  \adom(I) \cup C$.  Here, for instances $I$ with identifiers,
  $\adom(I)$ still denotes the set of data values from $\dom$
  (i.e., not identifiers) appearing in $I$.  Taking into account
  Remark~\ref{remeqid}, we can view $M$ as an injective function
  that maps $\ideq$-classes of instances of $\Kone$ to
  $\ideq$-classes of instances of $\Ktwo$.

  Now the crucial observation is that
  for any finite $D \subseteq \dom$,  the number of classes
  $[I]$ in $\Kone$ with $\adom(I) \subseteq D$ is infinite.
  Indeed, a value in the second column can be attached to one identifier,
  or to two, or to three, and so on.  In
  contrast, the number of classes $[J]$ in $\Ktwo$ with $\adom(J)
  \subseteq D \cup C$ is finite.  Indeed, every value in the second column
  is attached to exactly one identifier, so the only degree of
  freedom we have is how many distinct identifiers there are in
  the instance.  Hence, an function $M$ cannot be injective; a
  contradiction.

  For the other direction, suppose, for the sake of
  contradiction, that $\Ktwo \leqidgen \Kone$ by
  schema mapping $M$.  Again $M$ is $C$-generic for values for
  some finite $C \subseteq \dom$.  Pick four distinct elements
  $a$, $b$, $c$, $d$ in $\dom$ outside $C$, and consider the
  instance $I$ of $\Ktwo$ where \[ I(R)=\{(\id 1,a), (\id 1,b),
  (\id 2,c), (\id 2, d)\}. \]  As mentioned at the end of
  Section~\ref{secidgen}, we can use permutations of $\dom$ as
  isomorphisms of $\ideq$-classes of instances.  We observe that
  the permutation $(a\ c)(b\ d)$ is a $C$-automorphism of $[I]$,
  (the $\ideq$-class of $I$),
  but the individual transpositions $(a\ c)$ and $(b\ d)$ are
  not.

  We can also analyze
  how the automorphism group of instances $J$ of $\Kone$ must
  look like.  For any data value $x$ in
  the second column, we can count the number of
  identifiers that $x$ is attached to in $J$.  Call two data
  values $x$ and $y$ \emph{count-equivalent} if they have the same such
  count.  Then the automorphism group of $[J]$ is generated by all cyclic
  permutations of count-equivalent values.  In particular, if an
  automorphism is a composition of two disjoint cycles, then both
  cycles are themselves automorphisms.

  Now let $J_0$ be an instance of $\Kone$
  such that $M(I,J_0)$.  Genericity and injectivity
  of $M$ implies that $[J_0]$ has exactly the same
  $C$-automorphisms as $[I]$.  In particular, $(a\ c)(b\ d)$ is
  an automorphism of $[J_0]$.  But then also $(a\ c)$ must be
  such an automorphism.  However, it is not an automorphism of
  $[I]$; a contradiction.
\end{proof}

\begin{remark}
  To conclude this section we remark that schema $\Kboth$ is
  generically equivalent (with identifiers) to the simple
  vocabulary $\V_{11}$ with a single unary relation name $S$ with a value
  column (formally, of arity $(0,1)$).  Indeed, intuitively,
  the identifier column is redundant in $\Kboth$.
  Formally, $\Kboth \leqidgen \V_{11}$
  is realized simply by the projection on
  the second column of $R$. The converse
  $\V_{11} \leqidgen \Kboth$ is realized by the schema mappings
  that adds an identifier column to $S$, consisting of all
  distinct identifiers.
\end{remark}

\section{Conclusion and outlook} \label{seconc}

We have revived the systematic investigation of information
capacity of relational schemas.  Indeed, Hull \cite{hull_siam}
posed as an open question to characterize generic dominance,
which we have now done for the case of a single binary relation,
keys, and inclusion dependencies.  Moreover, we peeked behind the
curtain of the ternary case; we extended the basic notions to
accommodate the creation of identifiers; and we showed that these
extended notions are workable by applying them to the case of a
single binary relation with one identifier column and one value
column.

As already mentioned in the Introduction, the challenge is to map
out the information capacities of graph database schemas, and
compare these to common relational database schemas.  We think of
our work as a first step towards this challenge, by bringing
together all the necessary ingredients.  Thus, the next step is
to consider schemas with \textbf{multiple} unary and binary relations,
with identifiers, and keys and inclusion dependencies.  Indeed,
such schemas can already capture many aspects of graph schemas,
by modeling nodes and relationships using identifiers, and using
data values for properties \cite{kger}.

A further step would be to include more powerful key and
participation constraints that involve joins of the relations, as
proposed, e.g., in PG-Schema and KGER  \cite{pgschema,kger}.

Finally, we note that a number of further questions have remained
open since the work of Hull \cite{hull_siam}.  The most
interesting one concerns calculous dominance, mentioned in
Section~\ref{secprelim}.  It is indeed highly desirable, when
generic dominance holds, to be able to realize it by relational
calculus/algebra operations.  Indeed, in the results we have
presented here, this happens to hold.  An interesting question
is to characterize the situations where generic dominance and
calculous dominance coincide.

\bibliography{database-do-not-edit,extra} % JANVDB DO NOT EDIT DATABASE.BIB; use EXTRA.BIB

\begin{thebibliography}{10}

\bibitem{ahv_book}
S.~Abiteboul, R.~Hull, and V.~Vianu.
\newblock {\em Foundations of Databases}.
\newblock Addison-Wesley, 1995.

\bibitem{abihull_restructtcs88}
S.~Abiteboul and Richard Hull.
\newblock Restructuring hierarchical database objects.
\newblock {\em Theoretical Computer Science}, 62:3--38, 1988.

\bibitem{ak_iql}
S.~Abiteboul and P.C. Kanellakis.
\newblock Object identity as a query language primitive.
\newblock {\em Journal of the ACM}, 45(5):798--842, 1998.

\bibitem{shacl-shex-pgschema}
S.~Ahmetaj, I.~Boneva, J.~Hidders, K.~Hose, M.~Jakubowski, J.E. Labra~Gayo,
  W.~Martens, F.~Mogavero, F.~Murlak, C.~Okulmus, A.~Polleres, O.~Savkovic,
  M.~Simkus, and D.~Tomaszuk.
\newblock Common foundations for {SHACL}, {ShEx}, and {PG-Schema}.
\newblock In G.~Long et~al., editors, {\em Proceedings of the Web Conference},
  pages 8--21. ACM, 2025.

\bibitem{au_universality}
A.V. Aho and J.D. Ullman.
\newblock Universality of data retrieval languages.
\newblock In {\em Conference Record, 6th ACM Symposium on Principles of
  Programming Languages}, pages 110--120, 1979.

\bibitem{pgschema}
R.~Angles et~al.
\newblock {PG-Schema}: Schemas for property graphs.
\newblock {\em Proceedings of the ACM on Management of Data},
  1(2):198:1--198:25, 2023.

\bibitem{foundations-data-exchange-book}
M.~Arenas, P.~Barcel\'o, L.~Libkin, and F.~Murlak.
\newblock {\em Foundations of Data Exchange}.
\newblock Cambridge University Press, 2014.

\bibitem{atzeni-inclusion}
P.~Atzeni, G.~Ausiello, and C.~Batini.
\newblock Inclusion and equivalence between relational database schemata.
\newblock {\em Theoretical Computer Science}, 19(3):267--285, 1982.

\bibitem{bmsu_equivalence}
C.~Beeri, A.O. Mendelzon, Y.~Sagiv, and J.D. Ullman.
\newblock Equivalence of relational database schemes.
\newblock {\em SIAM Journal on Computing}, 10(2):352--370, 1981.

\bibitem{bonivoigt-graph-schema}
A.~Bonifati, P.~Furniss, A.~Green, R.~Harmer, E.~Oshurko, and H.~Voigt.
\newblock Schema validation and evolution for graph databases.
\newblock In A.H.F. Laender, B.~Pernici, et~al., editors, {\em Proceedings 38th
  International Conference on Conceptual Modeling}, volume 11788 of {\em
  Lecture Notes in Computer Science}, pages 338--456. Springer, 2019.

\bibitem{codd2}
E.~Codd.
\newblock Further normalization of the data base relational model.
\newblock In R.~Rustin, editor, {\em Data Base Systems}, pages 33--64.
  Prentice-Hall, 1972.

\bibitem{enrico-conceptual}
E.~Franconi and Th. Abgrall.
\newblock Logical foundations of conceptual modelling for relational and {SQL}
  databases: An introduction.
\newblock In C.M. Fonseca et~al., editors, {\em Advances in Conceptual
  Modeling: Proceedings ER 2025 Workshops}, volume 16190 of {\em Lecture Notes
  in Computer Science}, pages 5--25. Springer, 2026.

\bibitem{kger}
E.~Franconi, B.~Groz, J.~Hidders, N.~Pardal, S.~Staworko, J.~Van~den Bussche,
  and P.~Wieczorek.
\newblock The {KG-ER} conceptual schema language.
\newblock arXiv:2508.02548, 2025.

\bibitem{hull_siam}
R.~Hull.
\newblock Relative information capacity of simple relational schemata.
\newblock {\em SIAM Journal on Computing}, 15(3):856--886, 1986.

\bibitem{hy_format}
R.~Hull and C.K. Yap.
\newblock The format model, a theory of database organization.
\newblock {\em Journal of the ACM}, 31(3):518--537, 1984.

\bibitem{st1}
I.~Kobayashi.
\newblock Losslessness and semantic correctness of database schema
  transformation: Another look of schema equivalence.
\newblock {\em Information Systems}, 11(1):41--59, 1986.

\bibitem{st5}
P.~McBrien and A.~Poulovassilis.
\newblock Data integration by bi-directional schema transformation rules.
\newblock In {\em Proceedings 19th ICDE}, pages 227--238, 2003.

\bibitem{st2}
R.J. Miller, Y.E. Ioannidis, and R.~Ramakrishnan.
\newblock The use of information capacity in schema integration and
  translation.
\newblock In {\em Proceedings 19th VLDB}, pages 120--133, 1993.

\bibitem{odunlaing}
C.~O'Dunlaing and C.-K. Yap.
\newblock Generic transformations of data structures.
\newblock In {\em Proceedings 23rd Annual Symposium on Foundations of Computer
  Science}, pages 186--195. IEEE Computer Science Society, 1982.

\bibitem{oeis}
{OEIS Foundation Inc.}
\newblock The {O}n-{L}ine {E}ncyclopedia of {I}nteger {S}equences, 2025.
\newblock Published electronically at \url{http://oeis.org}.

\bibitem{st4}
A.~Poulovassilis and P.~McBrien.
\newblock A general formal framework for schema transformation.
\newblock {\em Data \& Knowledge Engineering}, 28(1):47--71, 1998.

\bibitem{st3}
X.~Qian.
\newblock Correct schema transformations.
\newblock In {\em Proceedings 5th EDBT}, pages 114--128, 1996.

\bibitem{tarski_logical}
A.~Tarski.
\newblock What are logical notions?
\newblock {\em History and Philosophy of Logic}, 7:143--154, 1986.
\newblock Edited by J. Corcoran.

\bibitem{ullman}
J.D. Ullman.
\newblock {\em Principles of Database and Knowledge-Base Systems}, volume~I.
\newblock Computer Science Press, 1988.

\bibitem{vvag_compl}
J.~Van~den Bussche, D.~Van~Gucht, M.~Andries, and M.~Gyssens.
\newblock On the completeness of object-creating database transformation
  languages.
\newblock {\em Journal of the ACM}, 44(2):272--319, 1997.

\bibitem{vanlint-wilson}
J.H. van Lint and R.M. Wilson.
\newblock {\em A Course in Combinatorics}.
\newblock Cambridge University Press, 1992.

\end{thebibliography}

\appendix

\section{Proof of Theorem~\ref{schroderbernstein}}
\label{appschroderbernstein}

We start with recalling a well-known proof for the Schr\"oder-Bernstein theorem that states that for any two sets, say $X$ and $Y$, for which there are injections $f : X \to Y$ and $g : Y \to X$ there exists a bijection $h : X \to Y$. 
    Consider the sets $X' = \{ 1 \} \times X$ and $Y' = \{ 2 \} \times Y$. Observe that since $f$ and $g$ are injections there are injections $f' : X' \to Y'$ defined by $f'(1, x) = (2, f(x))$ and an injection $g' : Y' \to X'$ defined by $g'(2, x) = (1, g(x))$. 
    We can represent the functions $f'$ and $g'$ together in a single directed bipartite graph $(V, E)$ where $V = X' \cup Y'$ and $E = \{ (x', y') \in V \times V \mid y' = f'(x') \vee y' = g'(x') \}$. This graph will be a disjoint union of connected components of one of the following shapes:
    \begin{center}        
 \begin{tikzpicture}[
  element/.style = {shape=circle, draw, line width=1, inner sep=0, minimum width=4},
  anode/.style = {element, fill=black},
  bnode/.style = {element, fill=white}
  ]
    
  \begin{scope}
    \node (a1) [anode] {};
    \node (b1) [bnode, right=0.5 of a1] {} edge[latex-] node [above] {\tiny $f'$} (a1);    
    \node (a2) [anode, right=0.5 of b1] {} edge[latex-] node [above] {\tiny $g'$} (b1);
    \node (b2) [bnode, right=0.5 of a2] {} edge[latex-] node [above] {\tiny $f'$} (a2);   
    \node (a3) [anode, right=0.5 of b2] {} edge[latex-] node [above] {\tiny $g'$} (b2);
    \node (b3) [bnode, right=0.5 of a3] {} edge[latex-] node [above] {\tiny $f'$} (a3);   
    \node (a4) [right=0 of b3] {$\ldots$};
    \node at (-1, 0) {(SH1)};
  \end{scope}    

  \begin{scope}[shift={(0, -0.7)}]
    \node (bb1) [bnode] {};
    \node (aa1) [anode, right=0.5 of bb1] {} edge[latex-] node [above] {\tiny $g'$} (bb1);    
    \node (bb2) [bnode, right=0.5 of aa1] {} edge[latex-] node [above] {\tiny $f'$} (aa1);
    \node (aa2) [anode, right=0.5 of bb2] {} edge[latex-] node [above] {\tiny $g'$} (bb2);   
    \node (bb3) [bnode, right=0.5 of aa2] {} edge[latex-] node [above] {\tiny $f'$} (aa2);
    \node (aa3) [anode, right=0.5 of bb3] {} edge[latex-] node [above] {\tiny $g'$} (bb3);   
    \node (bb4) [right=0 of aa3] {$\ldots$};
    \node at (-1, 0) {(SH2)};    
  \end{scope}    

  \begin{scope}[shift={(0, -1.4)}]
    \node (cb1) [] {$\dots$};
    \node (ca1) [anode, right=0 of cb1] {};    
    \node (cb2) [bnode, right=0.5 of ca1] {} edge[latex-] node [above] {\tiny $f'$} (ca1);
    \node (ca2) [anode, right=0.5 of cb2] {} edge[latex-] node [above] {\tiny $g'$} (cb2);   
    \node (cb3) [bnode, right=0.5 of ca2] {} edge[latex-] node [above] {\tiny $f'$} (ca2);
    \node (ca3) [anode, right=0.5 of cb3] {} edge[latex-] node [above] {\tiny $g'$} (cb3);   
    \node (cb4) [bnode, right=0.5 of ca3] {} edge[latex-] node [above] {\tiny $f'$} (ca3);       
    \node (ca4) [right=0 of cb4] {$\ldots$};
    \node at (-1, 0) {(SH3)};    
  \end{scope}

  \begin{scope}[shift={(0.5, -2.1)}]
    \node (db1) [bnode] {};
    \node (da1) [anode, right=0.5 of db1] {} edge[latex-]  node [above] {\tiny $g'$} (db1);    
    \node (de1) [right=0 of da1] {$\ldots$};
    \node (da2) [anode, right=0 of de1] {};   
    \node (db2) [bnode, right=0.5 of da2] {} edge[latex-]  node [above] {\tiny $f'$} (da2);
    \node (da3) [anode, below=0.5 of db2] {} edge[latex-, out=0, in=0, looseness=2] node [right] {\tiny $g'$} (db2);   
    \node (db3) [bnode, left=0.5 of da3] {} edge[latex-] node [above] {\tiny $f'$} (da3);       
    \node (de2) [left=0 of db3] {$\ldots$};
    \node (db4) [bnode, left=0 of de2] {};   
    \node (da4) [anode, left=0.5 of db4] {} edge[latex-] node [above] {\tiny $g'$} (db4) edge[-latex, looseness=2, out=180, in=180] node [left] {\tiny $f'$} (db1);       
     \node at (-1.5, -0.25) {(SH4)};       
  \end{scope}    
    
\end{tikzpicture}
       \end{center}
       Here the black and white nodes represent elements of $X'$ and $Y'$, respectively. The edges are labelled with the function they originate from. Note that every node has exactly one outgoing edge and at most one incoming edge. Given this it is not hard to see that any node in the graph is in a component of one of the four shapes. After all, if we consider a particular node $v$ and we start from it following a path of edges forwards we might encounter $v$ after a finite number of edges. If we do, $v$ is in shape SH4. If not, we can check what happens if we follow edges backwards from $v$. If this does not stop, $v$ is in shape SH3. If it does stop then $v$ is in shape SH1 if we end in a node from $X$, and in shape SH2 if we end in a node from $Y$. 

       Within a component we can construct a bijection from the black nodes to the white nodes as follows. For nodes in SH1, SH3 and SH4 we take the $f'$ edges, and for nodes in SH2 we take the inverse of the $g'$ edges. Since the components are disjoint, the union of all these bijections will be a bijection between all the black nodes and all the white nodes, and so a bijection $h' : X' \to Y'$.

       Based on $h'$ we define a function $h : X \to Y$ such that $h(x) = y$ iff $h'(1, x) = (2, y)$. Recall that $X' = \{ 1 \} \times X$ and $Y' = \{ 2 \} \times Y$, from which it is obvious that $h : X \to Y$ is a bijection if $h' : X' \to Y'$ is a bijection.

       We now return to our original setting, where the sets are $X = \sem{\schema_1}$ and $Y = \sem{\schema_2}$ with generic injections $f : X \to Y$ and $g : Y \to X$.
       We follow the preceding approach, so assume that $X' = \{ 1 \} \times \sem{\schema_1}$ and $Y' = \{ 2 \} \times \sem{\schema_2}$, then construct a bijection $h' : X' \to Y'$, from which we derive a bijection $h : X \to Y$. What remains to be shown is that $h$ is generic if $f$ and $g$ are generic.
       
       For a value permutation $\pi$ we define its analogue at the graph level as the function $\hat{\pi}' : V \to V$ such that $\hat{\pi}'(1, x) = (1, \hat{\pi}(x))$ and 
       $\hat{\pi}'(2, x) = (2, \hat{\pi}(x))$.      
       It is clear that $\hat\pi'$ is a bijection since $\hat\pi$ is a bijection. We can also observe that from the genericity of $f$ and $g$ it follows that we have this property at the graph level, i.e., $f' \circ \hat\pi' = \hat\pi' \circ f'$ and  $g' \circ \hat\pi' = \hat\pi' \circ g'$. We can use these observations to show that $\hat\pi'$ always maps nodes in a component of a particular shape to nodes in a component of the same shape. This is because $\hat\pi'$ defines a graph isomorphism, since it (i) defines a bijection from $V$ to $V$ and (ii) for any two nodes $v_1$ and $v_2$ it holds that the edge $(v_1, v_2)$ is in the graph iff the edge $(\hat\pi'(v_1), \hat\pi'(v_2))$ is in the graph. The latter holds because
       \begin{align*}
           (v_1, v_2) \in E & \Longleftrightarrow v_2 = f'(v_1) \vee v_2 = g'(v_1) \\
           & \Longleftrightarrow \hat\pi'(v_2) = \hat\pi'(f'(v_1)) \vee \hat\pi'(v_2) = \hat\pi'(g'(v_1)) \\
           & \Longleftrightarrow \hat\pi'(v_2) = f'(\hat\pi'(v_1)) \vee \hat\pi'(v_2) = g'(\hat\pi'(v_1)) \\          
           & \Longleftrightarrow (\hat\pi'(v_1), \hat\pi'(v_2)) \in E \text{ .}
       \end{align*}
       From this we can infer genericity for $h$ as follows. 
       First observe that if $f$ is $C$-generic and $g$ is $B$-generic, they will both be $(C \cup B)$-generic. We now set out to show that $h$ is also $(C \cup B)$-generic.
       Consider a $\valDom$ permutation $\pi$ that fixes all elements in $C \cup B$, and an instance $x \in X$. If $(1, x)$ is in a component of shape SH1, SH3 or SH4, then $h(x) = f(x)$. Moreover,  $\hat\pi'((1, x))$ is then in a component of the same shape and so  $h(\hat\pi(x)) = f(\hat\pi(x))$. Therefore $\hat\pi(h(x)) = \hat\pi(f(x)) = f(\hat\pi(x)) = h(\hat\pi(x))$. If, on the other hand, $(1, x)$ is in a component of shape SH2, then $h(x) = g^{-1}(x)$. Moreover,  $(1, \hat\pi(x))$ is then also in a component of shape SH2, so $h(\hat\pi(x)) = g^{-1}(\hat\pi(x))$. Therefore $\hat\pi(h(x)) = \hat\pi(g^{-1}(x)) = g^{-1}(\hat\pi(x)) = h(\hat\pi(x))$. So either way we can conclude that $\hat\pi(h(x)) = h(\hat\pi(x))$.

\section{Proof of Proposition~\ref{propcard}} \label{appcard}

For $\schema_0$
the cardinality function is $$\cfun_0 : n \mapsto 2^{n^2} \text{ .}$$
If $D$ has $n$ elements we can construct $n^2$ different pairs. Any subset of this is a valid instance, so there are $2^{n^2}$ possible instances.

For $\schema_1$ the cardinality function is $$\cfun_1 : n \mapsto \sum_{k=0}^n \binom{n}{k} (2^{k}-1)^k \text{ .}$$
We can first select a set of $k \leq n$ domain values from $D$ that will be $R[2]$. Then for we pick for each of them, so $k$ times, a non-empty subset, for which there are $2^k - 1$ options.

For $\schema_3$ the cardinality function is 
$$\cfun_3 : n \mapsto \sum_{k=0}^n \binom{n}{k} \mathbf{a}_k$$
where $\mathbf{a}_k$ is the number of $k \times k$ boolean matrices with no zero-row nor zero-column or, equivalently, the number of edge covers in the complete bipartite graph $\mathbf{K}_{k, k}$.

The formula describes the number of options for first picking the size of the $R[1] = R[2]$, and then it is not hard to see how each such matrix of that size corresponds to an instance of the schema and vice versa.

There is no known closed form for $\textbf{a}_k$, but it is known that
$$\textbf{a}_k = \sum_{i=0}^k (-1)^i\binom{k}{i}(2^{k-i}-1)^k .$$ For a reference see \cite[A048291]{oeis}. The proof of the formula can be done with the help of the inclusion-exclusion principle.

For $\schema_4$ the cardinality function is $$\cfun_4 : n \mapsto 2^{\binom{n}{2} + n} \text{ .}$$
Given $n$ nodes we have $\binom{n}{2}$ possible undirected edges with two distinct nodes, and $n$ possible loops. Any subset of these is in instance.

For $\schema_5$ the cardinality function is $$\cfun_5 : n \mapsto \sum_{k=0}^n \binom{n}{k} n^k = (n + 1)^n \text{ .}$$
We can first select a set $D'$ of $k \leq n$ domain values from $D$ that will be $R[1]$. Then we pick for each of them, so $k$ times, an element from $D$, for which there are $n$ options. The alternative equivalent formula can be explained as follows. For each element in the domain, we have $n+1$ choices: not use it in first column, or use it in first column together with any of the $n$ possible values for column $2$.

For $\schema_7$ the cardinality function is $$\cfun_7 : n \mapsto \sum_{k=0}^n \binom{n}{k} k^k \text{ .}$$
We can first select a set $D'$ of $k \leq n$ domain values from $D$ that will be $R[1]$. Then we pick for each element in $D'$, so $k$ times, an element from $D'$ for the second column, and for this there are $k$ options. Alternatively, we can explain this as first choosing $D'$ and the picking a function from $D'$ to $D'$, of which there are $k^k$.

For $\schema_{15}$ the cardinality function is $$\cfun_{15} : n \mapsto \sum_{k=0}^n \binom{n}{k}^2 k! \text{ .}$$
We first select a set $D'_1$ of $k \leq n$ domain values from $D$ for $R[1]$, and then a set $D'_2$ of $k$ domain values from $D$ for $R[2]$. Then we select a subset of $D'_1 \times D'_2$ that defines a bijection between $D'_1$ and $D'_2$, for which there are $k!$ possibilities. To see this, consider that we can order the elements in $D'_1$ and $D'_2$ in their natural order, and then each such bijection can be seen to correspond to a permutation of $k$ positions.

For $\schema_{6}$ the cardinality function is $$\cfun_{6} : n \mapsto \sum_{k=0}^n \binom{n}{k} k! \text{ .}$$
  The reasoning is as for $\sch_{15}$,
except here we first select a single set $D'$ of $k \leq n$ domain values from $D$ for $R[1] = R[2]$. Then we let $D'_1 = D'_2 = D'$ and pick the bijection, or which there are $k!$ choices.

%An interesting additional observation is that for $n > 0$ the
%formula can be simplified to $\lfloor n! e \rfloor$.

For $\schema_{9}$
the cardinality function is $$\cfun_{9} : n
\mapsto \sum_{k=0}^n \binom{n}{k} \mathbf{I}_k$$ where
$\mathbf{I}_k$ is the number of involutions on a set with $k$
elements.\footnote{An involution on $X$ is a function $f : X \to
X$ such that $f \circ f$ is the identity on $X$.}

  The reasoning is as for $\sch_{6}$,
except that we observe that the choice for the set of pairs corresponds to picking an involution. We are not aware of a known closed formula for $\mathbf{I}_k$ but the formula 
$$\mathbf{I}_k = \sum_{m=0}^{\lfloor \frac{k}{2} \rfloor} \frac{k!}{2^m m! (k-2m)!}$$
can be found in the literature and Wikipedia.

\section{Proof of Lemma~\ref{lemtriples}} \label{apptriples}

Let $R \in \sem{\T_1}$. We recall that $\key{R}{1}$ hence for
each $i\in\{2,3\}$ we can define a function $R[1,i]$ which maps
any $x \in R[1]$ to the value of component $i$ in the unique
triple $(x, y, z)\in R$, and maps any other $x$ to a default
value outside the domain, e.g., $\emptyset$.  We define the
mapping $f(R)$ as $\bigcup_{(x, y, z) \in R} f_0(x,y,z)$, where
$f_0$ is defined below.  To simplify the notations, all variables
which are not explicitly required to be equal on the right are
required to be distinct, which we represent by ``...''.  For each
$(x, y, z) \in R$ we define $f_0(x, y, z)$ as follows:

\[
\begin{array}{lll}
      \text{ (i)}& \{ (x, x, z)\} & \text{if  } x = y \ ... \\
 \text{ (i-a)}& \{ (x, x, x)\} &  \text{if  } x = y = z \ ... \\[.15cm]
      \text{ (ii)}& \{ (x, y, x)\} & \text{if  } x = z \ ...  \text{ and } x \neq R[1,3](y) \\
     \text{ (iii)}& \{ (x, y, x), (x, x, z) \} & \text{if  }  \ ...  \text{ and } x \neq R[1,3](y) \\
     \text{ (iii-a)}& \{ (x, y, x), (x, x, z) \} & \text{if  } y = z  \ ...  \text{ and } x \neq R[1,3](y) \\[.15cm]
     \text{ (iv)}& \{ (x, y, y), (x, x, z) \} & \text{if }  \ ...   \text{ and } x = R[1,3](y) \\
   \text{ (iv-a)}& \{ (x, y, y), (x, x, x) \} & \text{if } x = z \ ...  \text{ and } x = R[1,3](y) \\
\text{ (v)}& \{ (x, y, y) \} & \text{if } y = z \ ...  \text{ and } x = R[1,3](y) \\
\end{array}
\]

For any $(x, y, z)\in f(R)$, let $T_{x} = \{(x', y', z') \in f(R)
\mid x'=x\}$. 
By construction, the mapping preserves the first component of tuples, and this component is a key in $R$, hence for each $(x,y,z)\in R$,
we have $f_0(x, y, z) = T_x$.
To show that $f$ is injective, we must be able to
determine $y$ and $z$ from $T_x$ for all $x$.

We observe that (by considering whether $T_x$ contains $1$ or $2$
triples, then considering the equality types of the triples) we
can determine from $T_x$ the rule among cases (i) to (v) which
has been applied, because none of the $8$ sets of triples has
both the same number (1 or 2) and the same equality types of
triples.  Furthermore, in each case we can determine $y$ and $z$
from the set $T_x$ by considering which components are equal in
the triples. Consequently, $f$ is injective, and it is obviously
$\emptyset$-generic.

To prove that $f(R)$ belongs to $\T_{6}$, we first observe
that $f(R)$ satisfies $\key{f(R)}{1, 2}$ and $\key{f(R)}{1, 3}$.
This is because for each $x'\neq x$ the triples in $T_{x'}$
differ from those of $T_x$ on the first column, whereas when
$T_x$ contains two triples (cases (iii) to (iv-a)), those triples
differ on both the second and third column.  To validate our
construction, it remains to show that $f(R)$ also satisfies
$\key{f(R)}{2, 3}$, which is the cornerstone of our construction.
Let $\alpha, \alpha', \beta, \gamma$ be such that $(\alpha,
\beta, \gamma) \in f(R)$ and $(\alpha', \beta, \gamma) \in f(R)$. 

Let us first assume that $\beta = \gamma$.  Then the tuples must
have been generated by a case among (i-a), (iv), (iv-a) or (v).
If (i-a) has been used on, say, $\alpha$, then $\alpha = \beta =
\gamma$. Hence if (i-a) has been used on both $\alpha$ and
$\alpha'$ then $\alpha=\alpha'$. Otherwise we can assume w.l.o.g.
that  $\alpha$ has used one of (iv), (iv-a) or (v). Then $\alpha
= R[1,3](\beta)$ and $\key{R}{1}$ so if $\alpha'$ also has used
one of these rules, we also have $\alpha = R[1,3](\beta)$ hence
$\alpha' = \alpha$. For the same reason, if $\alpha'$ used rule
(i-a) then $\alpha' = \beta = R[1,3](\alpha')$ and therefore
$\alpha = R[1,3](\beta) = \alpha'$.  We have thus proved that
when $\beta = \gamma$, $\alpha'=\alpha$.

Let us now assume that $\beta \neq \gamma$.  Then the tuples must
have been generated by a case among (i), (ii), (iii), (iii-a) or
(iv) (and in the latter case only as $(x, x, z)$ and not as $(x,
y, y)$).  If the rules generating $\alpha$ and $\alpha'$ both
feature an $x$ in the second component, or both feature an $x$ in
the third component, then $\alpha=\alpha'$: this is the case when
we consider for instance 2 applications of (i), or when $\alpha$
is generated by (i) and $\alpha'$ by (iv).  The only remaining
possibility is therefore that one of the rules (w.l.o.g., the
rule generating $\alpha$) is one of (ii), (iii), or (iii-a)
generating a triple of the form $(x,y,x)$, i.e., a triple
$(\alpha, \beta, \gamma)$ such that $\alpha = \gamma$, and where
additionally $\alpha \neq R[1,3](\beta)$, whereas the other rule
is one of (i) or (iv) generating a triple $(\alpha', \beta,
\gamma)$ such that $\alpha' = \beta$ and (by construction)
$\gamma = R[1,3](\alpha')$.  Combining the two properties yields
a contradiction as $\alpha$ must simultaneously equal and differ
from $R[1,3](\alpha')$.  This concludes our proof that $f(R)$
satisfies $\key{f(R)}{2, 3}$ and therefore belongs to
$\T_{6}$.

\end{document}